\documentclass[preprint,12pt]{elsarticle}
\usepackage{amsthm}
\usepackage{amsmath}
\usepackage{amsfonts}
\usepackage{amssymb}
\usepackage{graphicx}
\usepackage[section]{algorithm}
\usepackage{algorithmic}
\usepackage{float}

\usepackage{graphicx}

\usepackage{amssymb}
\usepackage{amsmath}
\usepackage{amsthm}
\newtheorem{thm}{Theorem}[section]
\newtheorem{lem}{Lemma}[section]

\DeclareGraphicsExtensions{.eps,.jpg}

\begin{document}

\begin{frontmatter}



\title{An efficient parallel algorithm for the longest path problem in meshes}

\address[label1]{Corresponding author: fatemeh.keshavarz@aut.ac.ir
}
\address[label2]{ ar\_bagheri@aut.ac.ir}

\author{Fatemeh Keshavarz-Kohjerdi \fnref{label1}}
\address{Department of Computer Engineering,\\ Islamic Azad University,
    North Tehran Branch, Tehran, Iran.\\}

\author{Alireza Bagheri \fnref{label2}}

\address{Department of Computer Engineering \& IT,\\ Amirkabir University of Technology, Tehran, Iran.}

\begin{abstract}
In this paper, first we give a sequential linear-time algorithm for
the longest path problem in meshes. This algorithm can be considered
as an improvement of \cite{FAA:ALAFFLPIRGG}. Then based on this
sequential algorithm, we present a constant-time parallel algorithm
for the problem which can be run on every parallel machine.
\end{abstract}

\begin{keyword} grid graph, longest path
\sep meshes \sep sequential and parallel algorithms.

MSC: $05C45$; $05C85$; $05C38$.

\end{keyword}

\end{frontmatter}




\section{Introduction}\label{IntroSect}
The longest path problem, i.e. the problem of finding a simple path
with the maximum number of vertices, is one of the most important
problems in graph theory. The well-known NP-complete Hamiltonian
path problem, i.e. deciding whether there is a simple path that
visits each vertex of the graph exactly once, is a special case of
the longest path problem and has many applications \cite{D:GT,
GJ:CAI}. \par Only few polynomial-time algorithms are known for the
longest path problem for special classes of graphs. This problem for
trees began with the work of Dijkstra around 1960, and was followed
by other people \cite{BSZVGF:OCALPIAT, G:FALPIACMD, LMN:TLPPIPOIG,
MC:ASPAFTLPPOCG, UU:OCLPISGC}.
In the area of approximation
algorithms it has been shown that the problem is not in APX, i.e.
there is no polynomial-time approximation algorithm with constant
factor for the problem unless P=NP \cite{G:FALPIACMD}. Also, it has
been shown that finding a path of length $n-n^{\epsilon}$ is not
possible in polynomial-time unless P=NP \cite{KMR:OATLPIAG}. For the
backgrround and some known result about approximation algorithms, we
refer the reader to \cite{BH:FAPOSL, GN:FLPCAC, ZL:AFLPIG}.\par A
grid graph is a graph in which vertices lie only on integer
coordinates and edges connect vertices that are separated by a
distance of once. A solid grid graph is a grid graph without holes.
The rectangular grid graph $R(n,m)$ is the subgraph of $G^\infty$
(infinite grid graph) induced by $V(m,n)= \{\upsilon \ |\ 1 \leq
v_{x}\leq m, \ 1\leq v_{y}\leq n\}$, where $v_{x}$ and $v_{y}$ are
respectively $x$ and $y$ coordinates of $v$ (see Figure ~\ref{bb}).
A mesh $M(m,n)$ is a rectangular grid graph $R(m,n)$. Grid graphs
can be useful representation in many applications. Myers \cite{M}
suggests modeling city blocks in which street intersection are
vertices and streets are edges. Luccio and Mugnia \cite{LM:HPOARC}
suggest using a grid graph to represent a two-dimensional array type
memory accessed by a read/write head moving up, down or across. The
vertices correspond to the center of each cell and edges connect
adjacent cells. Finding a path in the grid corresponds to accessing
all the data.\par Itai \textit{et al.} \cite{IPS:HPIGG} have shown
that the Hamiltonian path problem for general grid graphs, with or
without specified endpoints, is NP-complete. The problem for
rectangular grid graphs, however, is in P requiring only
linear-time. Later, Chen \textit{et al.} \cite{CST:AFAFCHPIM}
improved the algorithm of \cite{IPS:HPIGG} and presented a parallel
algorithm for the problem in mesh architecture. There is a
polynomial-time algorithm for finding Hamiltonian cycle in solid
grid graphs \cite{LU:HCISGG}. Also, the authors in \cite{CT:HPOGG}
presented sufficient conditions for a grid graph to be Hamiltonian
and proved that all finite grid graphs of positive width have
Hamiltonian line graphs.\par Recently the Hamiltonian cycle (path)
and longest path problem of a grid graph has received much
attention. Salman \textit{et al.} \cite{AEB:SSAG} introduced a
family of grid graphs, i.e. alphabet grid graphs, and determined
classes of alphabet grid graphs that contain Hamiltonian cycles.
Islam \textit{et al.} \cite{Imnrx:hcihgg} showed that the
Hamiltonian cycle problem in hexagonal grid graphs is NP-complete.
Also, Gordon \textit{et al.} \cite{vyf:hpotgg} proved that all
connected, locally connected triangular grid graphs are Hamiltonian,
and gave a sufficient condition for a connected graph to be fully
cycle extendable and also showed that the Hamiltonian cycle problem
for triangular grid graphs is NP-complete.\par Moreover, Zhang and
Liu \cite{wqz} gave an approximation algorithm for the longest path
problem in grid graphs and their algorithm runs in quadratic time.
Also the authors in \cite{FAA:ALAFFLPIRGG} has been studied the
longest path problem for rectangular grid graphs and their algorithm
is based on divide and conquer technique and runs in linear time.
Some results of the grid graphs are investigated in
\cite{kb:hpiscogg, 6}. \par In this paper, we present a sequential
and a parallel algorithms for finding longest paths between two
given vertices in rectangular grid graphs (meshes). Our algorithm
has improved the previous algorithm \cite{FAA:ALAFFLPIRGG} by
reducing the number of partition steps from $O(m+n)$ to only a
constant.
\par The organization of the paper as follow: In Section 2, we review some necessary definitions
and results that we will need. A sequential algorithm for the
longest path problem is given in Section 3. In Section 4, a parallel
algorithm for the problem is introduced which is based on the
mentioned sequential algorithm. Conclusions is given in Section 5.
\begin{figure}[tb]
  \centering
  \includegraphics[scale=1]{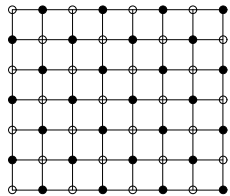}
  \caption[]%
  {\small The rectangular grid graph $R(8,7)$.}
\label{bb}
\end{figure}
\section{Preliminary results}
In this section, we give a few definitions and introduce the
corresponding notations. We then gather some previously established
results on the Hamiltonian and the longest path problems
in grid graphs which have been presented in \cite{CST:AFAFCHPIM, IPS:HPIGG, FAA:ALAFFLPIRGG}.\\
The \textit{two-dimensional integer grid} $G^\infty$ is an infinite
graph with vertex set of all the points of the Euclidean plane with
integer coordinates. In this graph, there is an edge between any two
vertices of unit distance. For a vertex $v$ of this graph, let
$v_{x}$ and $v_{y}$ denote $x$ and $y$ coordinates of its
corresponding point (sometimes we use $(v_x,v_y)$ instead of $v$).
We color the vertices of the two-dimensional integer grid as black
and white. A vertex $\upsilon$ is colored \textit{white} if
$\upsilon_{x}+\upsilon_{y}$ is even, and it is colored
\textit{black} otherwise. A \textit{grid graph} $G_{g}$ is a finite
vertex-induced subgraph of the two-dimensional integer grid. In a
grid graph $G_{g}$, each vertex has degree at most four. Clearly,
there is no edge between any two vertices of the same color.
Therefore, $G_{g}$ is a bipartite graph. Note that any cycle or path
in a bipartite graph alternates between black and white vertices. A
\textit{rectangular grid graph} $R(m,n)$ (or $R$ for short) is a
grid graph whose vertex set is $V(R)= \{\upsilon \ |\ 1 \leq
\upsilon_{x}\leq m, \ 1\leq \upsilon_{y}\leq n\}$. In the figures we
assume that $(1,1)$ is the coordinates of the vertex in the upper
left corner. The size of $R(m,n)$ is defined to be $mn$. $R(m,n)$ is
called \textit{odd-sized} if $mn$ is odd, and it is called
\textit{even-sized} otherwise. In this paper without loss of
generality, we assume $m\geq n$ and all rectangular grid graphs
considered here are odd$\times$odd, even$\times$odd and
even$\times$even. $R(m,n)$ is called a \textit{n-rectangle}.\\
The following lemma states a result about the Hamiltonicity of
even-sized rectangular graphs.
\begin{lem}
\label{Lemma:1} \cite{CST:AFAFCHPIM} $R(m,n)$ has a Hamiltonian
cycle if and only if it is even-sized and $m,n>1$.
\end{lem}

\begin{figure}[tb]
  \centering
  \includegraphics[scale=1]{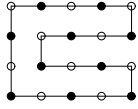}
  \caption[]%
  {\small A Hamiltonian cycle for the rectangular grid graph $R(5,4)$.}
\label{fig:hamcycle}
\end{figure}

\par Figure \ref{fig:hamcycle} shows a
Hamiltonian cycle for an even-sized rectangular grid graph, found by
Lemma \ref{Lemma:1}. Every Hamiltonian cycle found by this lemma
contains all the boundary edges on the three sides of the
rectangular grid graph. This shows that for an even-sized
rectangular graph $R$, we can always find a Hamiltonian cycle, such
that it contains all the boundary edges, except of exactly one side of $R$ which contains an even number of vertices.\par
Two different vertices $\upsilon$ and $\upsilon'$ in $R(m,n)$ are
called \textit{color-compatible} if either both $\upsilon$ and
$\upsilon'$ are white and $R(m,n)$ is odd-sized, or $\upsilon$ and
$\upsilon'$ have different colors and $R(m,n)$ is even-sized. Let
$(R(m,n),s,t)$ denote the rectangular grid graph $R(m,n)$ with two
specified distinct vertices $s$ and $t$. Without loss of generality,
we assume $s_{x} \leq t_{x}$.\\ $(R(m,n),s,t)$ is called
\textit{Hamiltonian} if there exists a Hamiltonian path between $s$
and $t$ in $R(m,n)$. An even-sized rectangular grid graph contains
the same number of black and white vertices. Hence, the two
end-vertices of any Hamiltonian path in the graph must have
different colors. Similarly, in an odd-sized rectangular grid graph
the number of white vertices is one more than the number of black
vertices. Therefore, the two end-vertices of any Hamiltonian path in
such a graph must be white. Hence, the color-compatibility of $s$
and $t$ is a necessary condition for $(R(m,n),s,t)$ to be
Hamiltonian. Furthermore, Itai \textit{et al.} \cite{IPS:HPIGG}
showed that if one of the following conditions hold, then
$(R(m,n),s,t)$ is not Hamiltonian:
\begin{itemize}
\item [(F1)] $R(m,n)$ is a 1-rectangle and either $s$ or $t$ is not a corner
vertex (Figure \ref{RecFig}(a))
\item[(F2)] $R(m,n)$ is a 2-rectangle
and $(s,t)$ is a nonboundary edge, i.e. $(s,t)$ is an edge and it is
not on the outer face (Figure \ref{RecFig}(b)).
\item [(F3)] $R(m,n)$ is
isomorphic to a 3-rectangle grid graph $R'(m,n)$ such that $s$ and
$t$ is mapped to $s'$ and $t'$ and all of the following three conditions hold:
\begin{enumerate}
\item $m$ is even,
\item $s'$ is black, $t'$ is white,
\item  $s'_{y}=2$ and $s'_{x}<t'_{x}$ (Figure \ref{RecFig}(c)) or
$s'_{y}\neq 2$ and $s'_{x}< t'_{x}-1$ (Figure \ref{RecFig}(d)).
\end{enumerate}
\end{itemize}
Also by \cite{IPS:HPIGG} for a rectangular graph $R(m,n)$ with two
distinct vertices $s$ and $t$, $(R(m,n),s,t)$ is Hamiltonian if and
only if $s$ and $t$ are color-compatible and $R(m,n)$, $s$ and $t$
do not satisfy any of conditions $(F1)$, $(F2)$ and $(F3)$.
In the following we use $P(R(m,n),s,t)$ to indicate the problem of
finding a longest path between vertices $s$ and $t$ in a rectangular
grid graph $R(m,n)$, $L(R(m,n),s,t)$ to show the length of longest
paths between $s$ and $t$ and $U(R(m,n),s,t)$ to indicate the upper
bound on the length of longest paths between $s$ and $t$.
\par The authors in
\cite{FAA:ALAFFLPIRGG} showed that the longest path problem between
any two given vertices $s$ and $t$ in rectangular grid graphs satisfies one of the
following conditions:
\begin{itemize}
\item [(C0)] $s$ and $t$ are color-compatible and none of (F1)- (F3) hold.
\item [(C1)] Neither (F1) nor (F$2^*$) holds and either
\begin{enumerate}
\item $R(m,n)$ is even-sized and $s$ and $t$ are same-colored or
\item $R(m,n)$ is odd-sized and $s$ and $t$ are different-colored.
\end{enumerate}
\item [(C2)]
\begin{enumerate}
\item $R(m,n)$ is odd-sized and $s$ and $t$ are black-colored and neither (F1) nor (F$2^*$) holds, or
\item $s$ and $t$ are color-compatible and (F3) holds.
\end{enumerate}
\end{itemize}
Where (F$2^*$) is defined as follows:
\begin{itemize}
\item [(F$2^*$)] $R(m,n)$ is a 2-rectangle and $s_{x}=\ t_{x}$ or $( s_{x} = t_{x}-1$ and $s_{y} \neq t_{y})$.
\end{itemize}
They also proved some upper bounds on the length of longest paths as
following:\\ $U(R(m,n),s,t)=
  \begin{cases}
    t_{x}-s_{x}+1,                               & \text{if  $(F1)$,} \\
    max(t_{x}+ s_{x},\ 2m - t_{x}- s_{x}+2), & \text{if $(F2^*)$,} \\
    mn, &           \text{if $(C0)$,}\\
    m n -1, &         \text{if $(C1)$,} \\
    mn-2, &         \text{if $(C2)$.}\\
  \end{cases}$
\begin{thm} \label{theorem:2}\cite{FAA:ALAFFLPIRGG} Let $U(R(m,n),s,t)$ be the upper bound on the length of
longest paths between $s$ and $t$ in $R(m,n)$ and let
$L(R(m,n),s,t)$ be the length of longest paths between $s$ and $t$.
In a rectangular grid graph $R(m,n)$, a longest path between any two
vertices $s$ and $t$ can be found in linear time and its length
$($i.e., $L(R(m,n),s,t))$ is equal to $U(R(m,n),s,t)$.
\end{thm}
\begin{figure}[tb]
  \centering
  \includegraphics[scale=1]{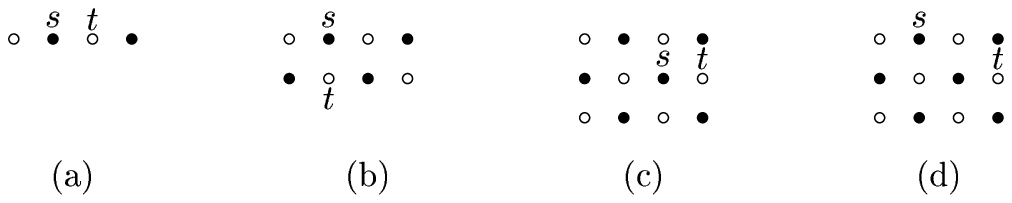}
  \caption[]%
  {\small Rectangular grid graph in which there is no Hamiltonian path between $s$ and $t$.}
  \label{RecFig}
\end{figure}

\section{\bf The sequential algorithm}
In this section, we present a sequential algorithm for finding a
longest path between two vertices in rectangular grid graphs. This
algorithm is the base of our parallel algorithm which is introduced
in Section 4. First, we solve the problem for 1-rectangles and
2-rectangles.
\begin{lem} \label{Lemma:5} \cite{FAA:ALAFFLPIRGG} Let $P(R(m,n),s,t)$ be a longest
path problem with $n=1$ or $n=2$, then
$L(R(m,n),s,t)=U(R(m,n),s,t)$.
\end{lem}
\begin{proof}
For a 1-rectangle obviously the lemma holds for the single possible
path between $s$ and $t$ (see Figure \ref{a3}(a)). For a
2-rectangle, if removing $s$ and $t$ splits the graph into two
components, then the path going through all vertices of the larger
component has the length equal to $U(R(m,n),s,t)$ (see Figure
\ref{a}(b)). Otherwise, let $s^{'}$ be the vertex adjacent to $s$
and $t^{'}$ be the vertex adjacent to $t$ such that $s^{'}_{y} \neq
s_{y}$  and $t^{'}_{y} \neq t_{y}$. Then we make a path from $s$ to
$s^{'}$ and a path from $t$ to $t^{'}$ as shown in Figure
\ref{a3}(c), (d), and connect $s^{'}$ to $t^{'}$ by a path such that
at most one vertex remains out of the path as depicted in this
figure.
\end{proof}
\begin{figure}[tb]
  \centering
  \includegraphics[scale=1]{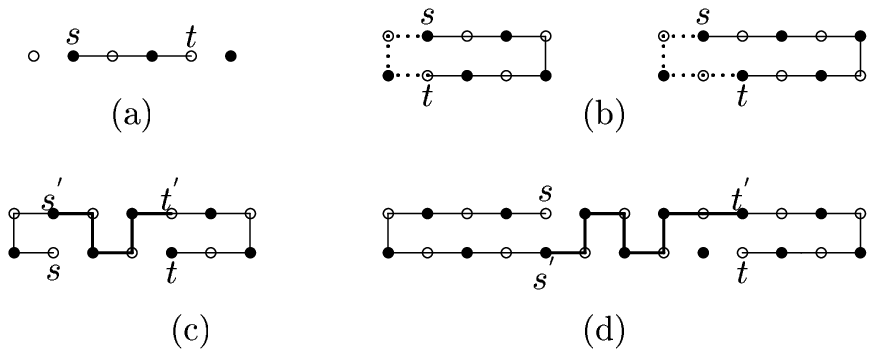}
  \caption[]%
  {
 (a) Longest path between $s$ and $t$ in a 1-rectangle,
  (b) Longest path between $s$ and $t$ in a 2-rectangle,
  (c) and (d) A path with length $2m$ and $2m-1$ for a 2-rectangle, respectively.}
  \label{a3}
\end{figure}
From now on, we assume that $m\geq n> 2$, so one of conditions (C0),
(C1) and (C2) should hold. Following the technique used in
\cite{CST:AFAFCHPIM} we develop an algorithm for finding longest
paths.\\

\noindent\textbf{Definition 3.1.} \cite{FAA:ALAFFLPIRGG} A
\textit{separation} of a rectangular grid graph $R$ is a partition
of R into two disjoint rectangular grid graphs $R_1$ and $R_2$, i.e.
$V(R)=V(R_{1})\cup V(R_{2})$, and $V(R_{1})\cap
V(R_{2})=\emptyset$.\\

\noindent\textbf{Definition 3.2.} \cite{IPS:HPIGG} Let $\upsilon$
and $\upsilon'$ be two distinct vertices in $R$. If
$\upsilon_{x}\leq2$ and $\upsilon'_{x}\geq m-1$, then $\upsilon$ and
$\upsilon'$ are called \textit{antipodes}.
\\

\noindent\textbf{Definition 3.3.} \cite{CST:AFAFCHPIM} Partitioning
a rectangular grid graph $R$ into five disjoint rectangular grid
subgraphs $R_{1}-R_{5}$ that is done by two horizontal and two
vertical separations are called \textit{peeling operation}, if the following two conditions hold:
\begin{enumerate}
\item $s,t\in R_{5}$ and $s$ and $t$ are antipodes.
\item  Each of four rectangular grid subgraphs $R_{1}-R_{4}$ is an
even-sized rectangular grid graph whose boundary sizes are both
greater than one, or is empty.
\end{enumerate}
Generally the two vertical separation of a peeling are done before
the two horizontal separation. However, for an odd$\times$odd or
odd$\times$even rectangular grid graph with $s_{x}=t_{x}$, this
order is reversed in order to guarantee that the boundary sizes of
$R_3$ and $R_4$ are greater than one. Figure
\ref{a} shows a peeling on $R(15,11)$ where $s$ is $(6, 5)$ and $t$ is $(8, 9)$.\\
\begin{figure}[htb]
  \centering
  \includegraphics[scale=1]{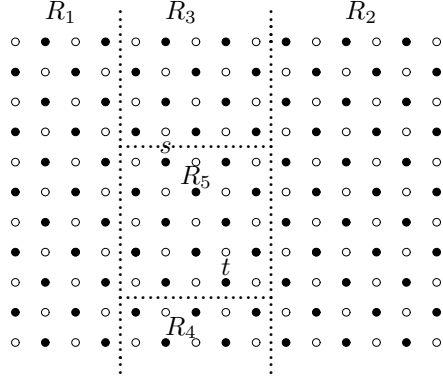}
  \caption[]%
  {\small A peeling on $R(15,11)$.}
  \label{a}
\end{figure}
The following lemma can be obtained directly from Definition 3.3.
\begin{lem} \label{Lemma:5} \cite{CST:AFAFCHPIM}
Let $R_{5}(n_{5},m_{5})$ be the resulting rectangular grid subgraph
of a peeling on $R(n,m)$, where $s,t \in V(R_{5})$. Then
\begin{enumerate}
\item $s,t$ remain the same color in $R_{5}$ as in $R$; and
\item $R_{5}$ has the same parity as $R$, that is, $m_{5}
\mod 2 = m \mod 2$, and $n_{5} \mod 2 = n \mod 2$.
\end{enumerate}
\end{lem}

\noindent\textbf{Definition 3.4.} A peeling operation on $R$ is
called \textit{proper} if
$|R_{1}|+|R_{2}|+|R_{3}|+|R_{4}|+U(R_{5},s,t) = U(R(n,m),s,t)$,
where $|R_{i}|$ denotes the number of vertices of
$R_{i}$.
\begin{lem} \label{Lemma:6} For the longest path problem $P(R(n,m),s,t)$, any peeling on $R(m,n)$ is proper if:
\begin{enumerate}
\item The condition $(C0)$ holds and $m\bmod 2 = n\bmod 2$ $($i.e. $R(m,n)$ is $even\times even$ or $odd \times odd)$, or
\item One of the conditions $(C1)$ and $(C2)$ hold and $R(m,n)$ is $even \times odd$ or $odd \times
odd$.
\end{enumerate}
\end{lem}
\begin{proof}
The lemma has been proved for the case that (C0) holds (see
\cite{CST:AFAFCHPIM}). So, we consider conditions (C1) and (C2).
From Lemma \ref{Lemma:5}, we know that $s$ and $t$ are still
color-compatible, and we are
going to prove that $P(R_5(m_5,n_5),s,t)$ is not in cases $F1$ and $F2^{*}$.\\
By Lemma \ref{Lemma:5}, when $R(m,n)$ is an $odd\times odd$
rectangular grid graph, $R_{5}(m_{5},n_{5})$ is also an $odd\times
odd$ rectangular grid graph, $s$ and $t$ have the same color as in
$R$, and hence $R_{5}(m_{5},n_{5})$ is not a 2-rectangle. If $R_{5}$
is a 1-rectangle, then $s_{y}=t_{y}$ or $s_{x}=t_{x}$ and then we
have the two following cases:\par Case1. (C1) holds and both $s$ and
$t$ are different color. In this case, one of $s_{x}$ and $t_{x}$
($s_{y}$ and $t_{y}$) is even and the other is odd. Considering that
$s$ and $t$ are antipodes and $R_5$ is odd$\times$odd, one of $s$
and $t$ must be at the corner and exactly one of the vertex goes out
of the path.
\par Case2. (C2) holds and both $s$ and $t$ are black color. In this case, all
$s_{x}$, $s_{y}$, $t_{x}$ and $t_{y}$ are even. Hence, vertices $s$
and $t$ are before corner vertices and exactly two vertices go out
of the path.\par In the similar way, when $R(m,n)$ is an $even\times
odd$ rectangular grid graph ((C1) holds), $R_{5}(m_{5},n_{5})$ is
also an $even \times odd$ rectangular grid graph, and hence
$R_{5}(m_{5},n_{5})$ is not a 2-rectangle. If $R_{5}$ is a
1-rectangle, then $s_{y}=t_{y}$. In this case, $s_{x}$ and $t_{x}$
are both odd or even. Hence, $s$ or $t$ are at the
corner and exactly one vertex goes out of the path. \\
Therefore by Theorem \ref{theorem:2},
$U(R_{5}(m_{5},n_{5}),s,t)=U(R(m,n),s,t)$ and any peeling of
$R(m,n)$ is always proper.
\end{proof}
Nevertheless, a peeling operation in an $even\times even$
rectangular grid graph $R(m,n)$ may not be proper, and
$U(R_{5}(m_{5},n_{5}),s,t)\neq U(R(m,n),s,t)$, see Figure \ref{bs}
where the dotted-lines represent a peeling operation. In the two
following cases a peeling operation is not proper:
\begin{enumerate}
\item [(F$1^{'}$)] $s$ is black, $s_x$ is even (or
odd), $t_y=s_y+1$ and $s_x\neq t_x$
\item [(F$2^{'}$)]$s$ is white, $s_x$ is even (or odd), $t_y=s_y-1$ and
$s_x\neq t_x$;
\end{enumerate}
\begin{figure}[tb]
  \centering
  \includegraphics[scale=1]{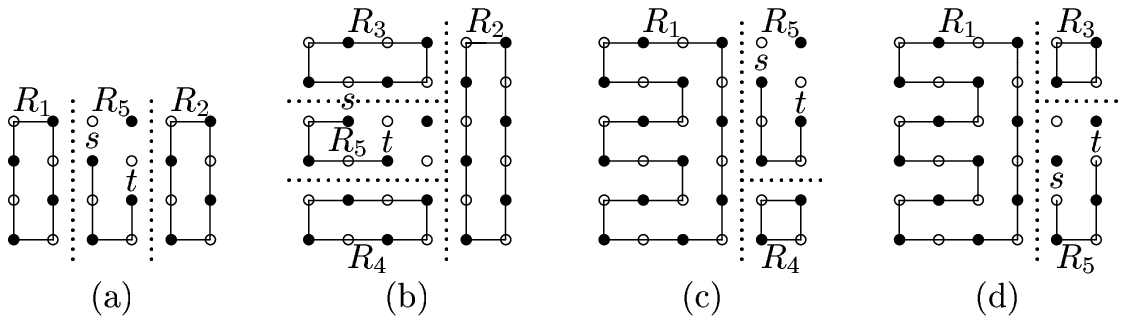}
  \caption[]%
  {\small Rectangular grid graph in which a peeling operation is not proper.}
  \label{bs}
\end{figure}
\begin{figure}[tb]
  \centering
  \includegraphics[scale=1]{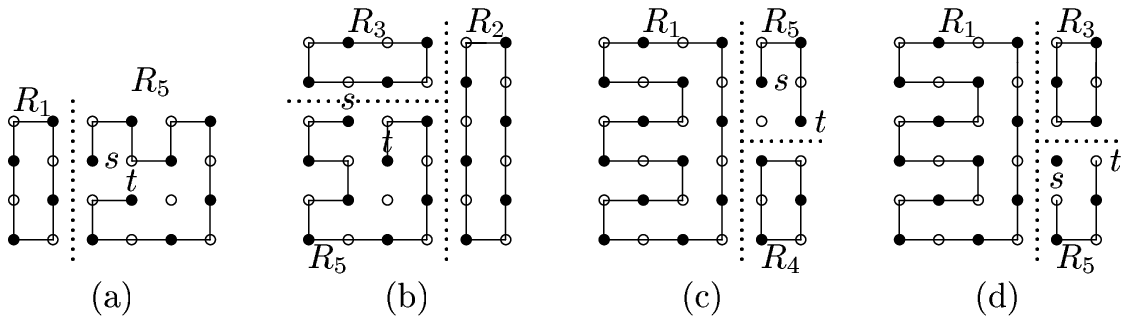}
  \caption[]%
  {\small }
  \label{bj}
\end{figure}
\begin{lem} \label{Lemma:7} For the longest path problem $P(R(n,m),s,t)$, where
$R(m,n)$ is an $even \times even$ rectangular grid graph, a peeling
operation on $R(m,n)$ is proper if and only if $P(R(n,m),s,t)$ is
not cases in $(F1^{'})$ and $(F2^{'})$.
\end{lem}
When a peeling operation is not proper it can be made proper by
adjustment the peeling boundaries. In that case, if $R_{1},\ R_{2},
\ R_{3}$ and $R_{4} $ are empty, then $ R_{5}$ is 2-rectangle that
is in case (F$2^*$). Therefore, without loss of generality, we
assume $R_{1},\ R_{2}, \ R_{3}$ or $R_{4} $ is not empty. If
rectangular grid subgraphs $R_{3}$ and $R_{4}$ are empty, then we
move one column (or two columns when $R_{1}$ or $R_{2}$ is a
2-rectangle) from $R_{1}$ or $R_{2}$ to $R_{5}$ such that $R_{1}$ or
$R_{2}$ is still even-sized rectangular grid graphs; see Figure
\ref{bj}(a). If $R_{1},\ R_{2}, \ R_{3}$ and $R_{4}$ (or $R_{3}$ and
$R_{4}$) is not empty, then we move one row (or two rows when
$R_{3}$ or $R_{4}$ is 2-rectangle) from $R_{3}$ or $R_{4}$ to
$R_{5}$ (Figure \ref{bj}(b)), or move the bottom row to $R_{4}$
(Figure \ref{bj}(c)) or move the upper row to $R_{3}$ (Figure
\ref{bj}(d)), such that $R_{3}$ or $R_{4}$ is still even-sized
rectangular grid graphs.
\par After a peeling operation on $R(m,n)$, we construct longest paths
in $R_5(m_5,n_5)$. Consider the following cases for $R_5(m_5,n_5)$:
\begin{itemize}
\item[(a)] $m_{5},n_{5}\leq 3$.
\item[(b)] $m_{5},n_{5}$ are even, and either $m_{5}\geq4$ or $n_{5}\geq4$;
\item[(c)] $m_{5},n_{5}$ are odd, and either $m_{5}\geq5$ or $n_{5}\geq5$;
\item[(d)] $m_{5}$ is even and $n_{5}$ is odd, and either $m_{5}\geq4$ or $n_{5}\geq5$.
\end{itemize}
For case (a), we showed that when $n=1,2$ the problem can be solved
easily. For $m,n=3$ the longest paths of all the possible problems
are depicted in Figure \ref{k} (the isomorphic cases are omitted).
 \begin{figure}
  \centering
\includegraphics[scale=1]{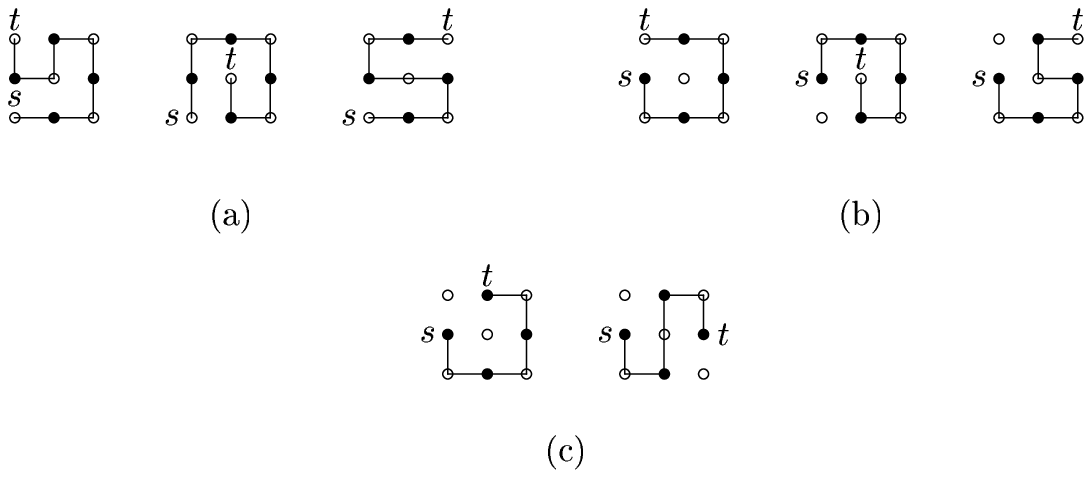}
  \caption[]%
  {\small For $n=m=3$, (a) $s$ and $t$ are white, then there is Hamiltonian path,
   (b) $s$ and $t$ have different colors, then there is a path with $U(R,s,t)=mn-1$
   and (c) $s$ and $t$ are black, then there is a path with $U(R,s,t)=mn-2$.}
  \label{k}
\end{figure}
\par For cases (b), (c) and (d) we use
the definition of trisecting.\\

\noindent\textbf{Definition 3.5.} \cite{CST:AFAFCHPIM} Two
separations of $R_{5}$ that partition it into three rectangular grid
subgraphs $R^{s}_{5}$, $R^{t}_{5}$ and $R^{m}_{5}$ is called
\textit{trisecting}, if
\begin{itemize}
\item[$(i)$.] $R^{s}_{5}$ and $R^{t}_{5}$ are a 2-rectangle, and
\item[$(ii)$.] $s \in V(R^{s}_{5})$ and $t \in V(R^{t}_{5})$.
\end{itemize}

A trisecting can be done by two ways horizontally and vertically. If
$m_{5}<4$ or $m_{5},n_{5}\geq 4$, then trisecting is done
horizontally, if $n_{5}<4$, then trisecting is done vertically. \\

\noindent\textbf{Definition 3.6.} A corner vertex on the boundary of
$R^{s}_{5}$ $($resp., $R^{t}_{5})$ facing $R^{m}_{5}$ is called a
\textit{junction vertex} of $R^{s}_{5}$ $($resp., $R^{t}_{5})$ if
either
\begin{itemize}
\item [(i)] The condition $(C0)$ holds and it has different color from $s$ and
$t$, or
\item [(ii)] One of the conditions $(C1)$ or $(C2)$ hold and
$U(R_{5}^{s},s,p)+U(R_{5}^{m},m,m^{'})+U(R_{5}^{t},q,t)=U(R(m,n),s,t)$.
Where $p$ is one of the corner vertices of $R^{s}_5$, $q$ is one of
the corner vertices of $R^{t}_5$, and $m$ and $m^{'}$ are two of the
corner vertices of $R^{m}_5$ facing $R^{s}_5$ and $R^{t}_5$,
respectively.
\end{itemize}

In Figure \ref{ca}, $p_1$ and $p_2$, $q_{1}$ and $q_{2}$, $m_{1}$,
$m_{2}$, $m_{3}$ and $m_{4}$ are junction vertices in $R^{s}_5$,
$R^{t}_5$ and $R^{m}_5$, respectively.
\begin{figure}[tb]
  \centering
  \includegraphics[scale=1]{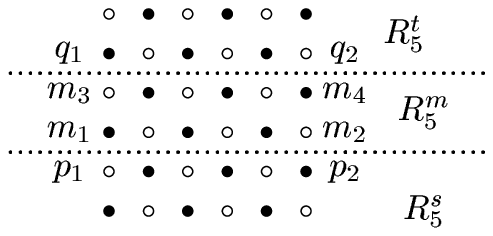}
  \caption[]%
  {\small A trisecting on $R(6,6)$.}
  \label{ca}
\end{figure}
Existence of junction vertices has been proved for condition $(C0)$
in \cite{CST:AFAFCHPIM}, in this paper we only consider conditions
$(C1)$ and $(C2)$.
\begin{lem} \label{Lemma:7} Performing a trisecting on $R_{5}$, where $m,n> 3$, assuming condition $(C1)$ or $(C2)$
holds, if $n_5=4$, and $s$ and $t$ facing the common border of
$R^{s}_{5}$ and $R^{t}_{5}$, then there is no junction vertex for
$R^{s}_{5}$ and $R^{t}_{5}$, otherwise $R^{s}_{5}$ and $R^{t}_{5}$
have at least one junction vertex.
\end{lem}
\begin{proof}
Consider Figure \ref{cc}(a) and (b), where $n_5=4$ and two vertices
$s$ and $t$ facing the common border $R^{s}_{5}$ and $R^{t}_{5}$. In
this case, the only two vertices $p_2$ and $q_1$ may be junction
vertices. By Theorem \ref{theorem:2}, there exists a Hamiltonian
path from $s$ to $p_2$ and from $q_1$ to $t$ in $R^{s}_{5}$ and
$R^{t}_{5}$, respectively, and
$U(R_{5}^{s},s,p_{2})+U(R_{5}^{t},q_{1},t)\neq U(R(m,n),s,t)$.
Hence, neither $R^{s}_{5}$ nor $R^{t}_{5}$ has a junction vertex at
all. Now for other cases, we show that $R^{s}_{5}$ and $R^{t}_{5}$
have at least one junction vertex.\\
In case (b), $s$ and $t$ have the same color (white or black), and
two corner vertices on the boundary of $R^{s}_{5}$ (resp.,
$R^{t}_{5}$) facing $R^{m}_{5}$ are different color and also
$R^{m}_{5}$ is a $k-$rectangle that $k$ is empty or $k\geq 2$ and
even. We consider the following three cases for $s$ and $t$:
\par Case 1. Both $s$ and $t$ are the corner vertices on the boundary
of $R^{s}_{5}$ and $R^{t}_{5}$ facing $R^{m}_{5}$; see Figure
\ref{cc}(c). By Theorem \ref{theorem:2}, there exists a Hamiltonian
path from $s$ to $p_{2}$ and from $q_{1}$ to $t$ in $R^{s}_{5}$ and
$R^{t}_{5}$, respectively, and a path from $m_{3}$ to $m_{2}$ which does not pass through a vertex in $R^{m}_{5}$. Therefore,
$U(R_{5}^{s},s,p_{2})+U(R_{5}^{m},m_{3},m_{2})+U(R_{5}^{t},q_{1},t)=U(R(m,n),s,t)$
and hence both $R^{s}_5$ and $R^{t}_5$ have a unique junction
vertex.
\par Case
2. $s$ is the corner vertex on the boundary of $R^{s}_{5}$ facing
$R^{m}_{5}$; see Figure \ref{cc}(d). By Theorem \ref{theorem:2},
there exists a Hamiltonian path from $s$ to $p_{2}$, from $q_{1}$ to
$t$ and a path from $m_3$ to $m_2$ which does not pass through a
vertex, or  a Hamiltonian path from $s$ to $p_{2}$, form  $m_{2}$ to $m_{4}$ and a path
from $q_{2}$ to $t$ which does not pass through a vertex. Therefore,
$U(R_{5}^{s},s,p_{2})+U(R_{5}^{m},m_{3},m_{2})+U(R_{5}^{t},q_{1},t)=U(R(m,n),s,t)$
and
$U(R_{5}^{s},s,p_{2})+U(R_{5}^{m},m_{2},m_{4})+U(R_{5}^{t},q_{2},t)=U(R(m,n),s,t)$
and hence $R^{s}_5$ has a unique junction vertex and $R^{t}_5$ have
two junction vertices (the same argument is also applied to $t$). In
this case, where $n_5=4$, both $R^{s}_5$ and $R^{t}_5$ have a unique
junction vertex; see Figure \ref{cc}(e).
\par Case 3. $s$ and $t$ are not the corner vertices on the
boundary of $R^{s}_{5}$ and $R^{t}_{5}$ facing $R^{m}_{5}$; see
Figure \ref{cc}(f). By Theorem \ref{theorem:2}, there exists a
Hamiltonian path from $s$ to $p_{1}$, $m_1$ to $m_3$ and a path from
$q_1$ to $t$ which does not pass through a vertex, or a Hamiltonian path from $s$ to
$p_{1}$, from $q_2$ to $t$ and a path from $m_1$ to $m_4$ which does
not pass through a vertex, or a Hamiltonian path form $m_2$ to $m_4$, from $q_2$ to $t$
and a path from $s$ to $p_{2}$ which does not pass through a vertex.
Therefore,
$U(R_{5}^{s},s,p_{1})+U(R_{5}^{m},m_{1},m_{3})+U(R_{5}^{t},q_{1},t)=U(R(m,n),s,t)$,
$U(R_{5}^{s},s,p_{1})+U(R_{5}^{m},m_{1},m_{4})+U(R_{5}^{t},q_{2},t)=U(R(m,n),s,t)$
and
$U(R_{5}^{s},s,p_{2})+U(R_{5}^{m},m_{2},m_{4})+U(R_{5}^{t},q_{2},t)=U(R(m,n),s,t)$
and hence both $R^{s}_{5}$ and $R^{t}_{5}$ have two junction
vertices.\\
\begin{figure}[tb]
  \centering
  \includegraphics[scale=1]{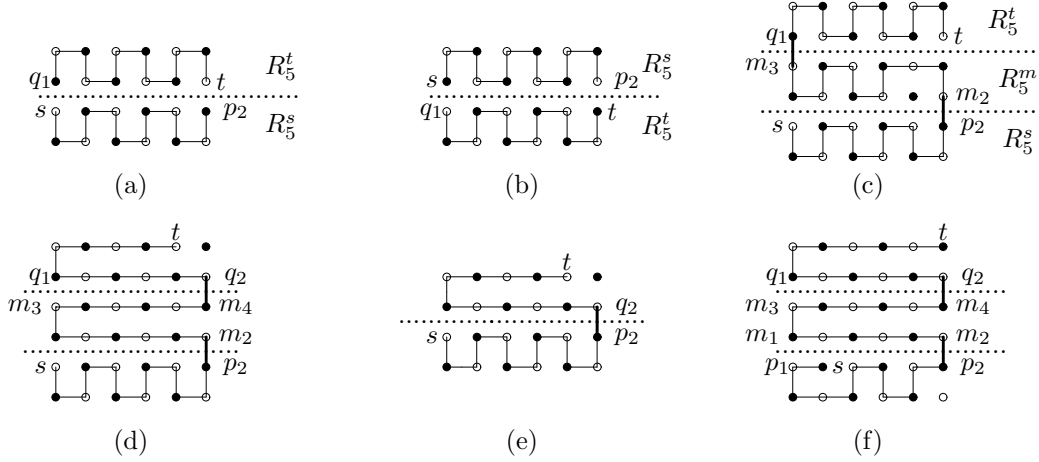}
  \caption[]%
  {\small A trisecting on $R(6,6)$, $R(6,4)$.}
  \label{cc}
\end{figure}
In case (c), $s$ and $t$ are black or different color, and two
corner vertices on the boundary of $R^{s}_{5}$ (resp., $R^{t}_{5}$)
facing $R^{m}_{5}$ are black and also $R^{m}_{5}$ is a $k-$rectangle
that $k\geq 1$ and odd. There are three cases for $s$ and $t$:
\par Case 1. Both $s$ and $t$ are the corner vertices on the boundary
of $R^{s}_{5}$ and $R^{t}_{5}$ facing $R^{m}_{5}$; see Figure
\ref{ccv}(a). Then $s$ and $t$ are black. By Theorem
\ref{theorem:2}, there exists a path from $s$ to $p_{2}$, from
$q_{1}$ to $t$ which does not pass through a vertex in $R^{s}_{5}$
and $R^{t}_{5}$, respectively, and a Hamiltonian path from $m_{3}$
to $m_{2}$ in $R^{m}_{5}$. Therefore,
$U(R_{5}^{s},s,p_{2})+U(R_{5}^{m},m_{3},m_{2})+U(R_{5}^{t},q_{1},t)=U(R(m,n),s,t)$
and hence both $R^{s}_5$ and $R^{t}_5$ have a unique junction
vertex.
\par Case 2. $s$ is the corner vertex on the boundary of $R^{s}_{5}$
facing $R^{m}_{5}$, then $s$ is black and $t$ is black or white; see
Figure \ref{ccv}(b). By Theorem \ref{theorem:2}, there exists a path
from $s$ to $p_{2}$ and from $q_{1}$ (or $q_{2}$) to $t$, where $t$
is black, which does not pass through a vertex, and a Hamiltonian
path from and $m_{2}$
to $m_{4}$ (or from $m_{3}$ to $m_{2}$), or a Hamiltonian path from $q_1$ (or $q_2)$ to $t$, where $t$ is white and $m_{2}$
to $m_{4}$ (or from $m_{3}$ to $m_{2}$) and a path
from $s$ to $p_{2}$ which does not pass through a vertex . Therefore,
$U(R_{5}^{s},s,p_{2})+U(R_{5}^{m},m_{3},m_{2})+U(R_{5}^{t},q_{1},t)=U(R(m,n),s,t)$
and
$U(R_{5}^{s},s,p_{2})+U(R_{5}^{m},m_{2},m_{4})+U(R_{5}^{t},q_{2},t)=U(R(m,n),s,t)$
and hence $R^{s}_5$ has a unique junction vertex and $R^{t}_5$ have
two junction vertices (the same argument is also applied to $t$). In
this case, where $n_5=5$, both $R^{s}_{5}$ and $R^{t}_{5}$ a unique
junction vertex; see Figure \ref{ccv}(c).
\par Case 3. $s$ and $t$ are not the corner vertices on the
boundary of $R^{s}_{5}$ and $R^{t}_{5}$ facing $R^{m}_{5}$; see
Figure \ref{ccv}(d). By Theorem \ref{theorem:2}, there exists a
Hamiltonian path from $s$ to $p$, $m$ to $m^{'}$ and $q$ to $t$,
where $s$ (or $t$) is white, and a path form $s$ to $p$ and $q$ to
$t$ which does not pass through a vertex where $s$ (or $t$) is
black, $p$ is $p_1$ or $p_2$, $q$ is $q_1$ or $q_2$, $m$ is $m_1$ or $m_2$ and $m^{'}$ is $m_3$ or $m_4$. Therefore,
$U(R_{5}^{s},s,p)+U(R_{5}^{m},m,m^{'})+U(R_{5}^{t},q,t)=U(R(m,n),s,t)$
and hence both $R^{s}_{5}$ and $R^{t}_{5}$ have two junction
vertices.\\
In case (d), if $n_{5} > 3$, the trisecting is performed
horizontally, and the claim is proved by applying the same argument
for case (b); see Figure \ref{ccv}(e). If $n_{5}=3$, the trisecting
is performed vertically
 and also
two corner vertices on the boundary of $R^{s}_{5}$ facing
$R^{m}_{5}$ are black. Therefore, the claim is proved by applying
the same argument for case (c).
\end{proof}

\begin{figure}[htb]
  \centering
  \includegraphics[scale=1]{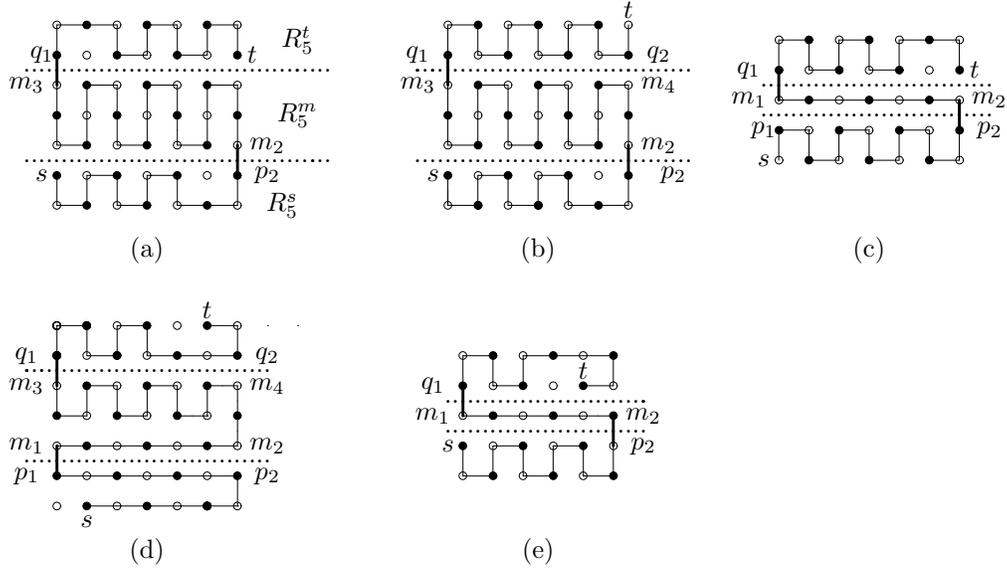}
  \caption[]%
  {\small A trisecting on $R(7,7)$, $R(7,5)$ and $R(6,5)$.}
  \label{ccv}
\end{figure}
After trisecting, we construct a longest path in $R^{s}_5$,
$R^{m}_5$ and $R^{t}_5$ between $s$ and $p$, $m$ and $m^{'}$ and $q$
and $t$, respectively. In the case that none of $R^{s}_{5}$ and
$R^{t}_{5}$ have junction vertices (when $n_{5}=4$ and both $s$ and
$t$ facing the common border $R^{s}_{5}$ and $R^{t}_{5}$), we
construct a longest path in $R^{s}_{5}$ (resp., $R^{t}_5$) between
$s$ (resp., $t)$ and a none-corner vertex of the boundary facing
$R^{t}_{5}$ (resp., $R^{s}_5)$; see Figure \ref{cd1}. At the end,
the longest paths in $R_5$ are combined through the junction
vertices; see Figures \ref{cc} and \ref{ccv}.\par Then we construct
Hamiltonian cycles in rectangular grid subgraphs $R_1$ to $R_4$, by
Lemma \ref{Lemma:1}; see Figure \ref{cd}. Then combine all
Hamiltonian cycles to a single Hamiltonian cycle. \par Two
non-incident edges $e_1$ and $e_2$ are parallel, if each end vertex
of $e_1$ is adjacent to some end vertex of $e_2$. Using two parallel
edges $e_1$ and $e_2$ of two Hamiltonian cycles (or a Hamiltonian
cycle and a longest path), such as two darkened edges of Figure
\ref{cm}(a), we can combine them as illustrated in Figure
\ref{cm}(b) and obtain a large Hamiltonian cycle.\par Combining the
resulted Hamiltonian cycle with
the longest path of $R_5$ is done as in Figure \ref{cn}. \\

\begin{figure}[htb]
  \centering
  \includegraphics[scale=1]{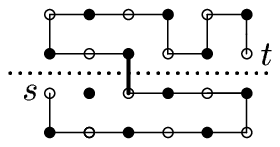}
  \caption[]%
  {\small the Longest path in $R(6,4)$.}
  \label{cd1}
\end{figure}

\begin{figure}[htb]
  \centering
  \includegraphics[scale=1]{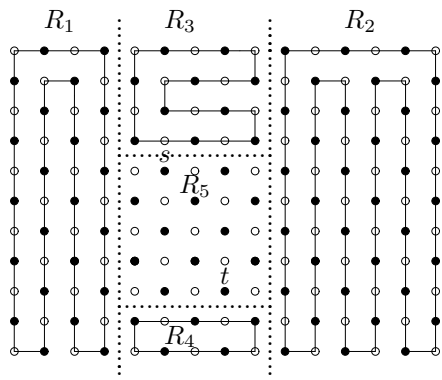}
  \caption[]%
  {\small Hamiltonian cycles in $R_1$ to $R_4$.}
  \label{cd}
\end{figure}

\begin{figure}[htb]
  \centering
  \includegraphics[scale=1]{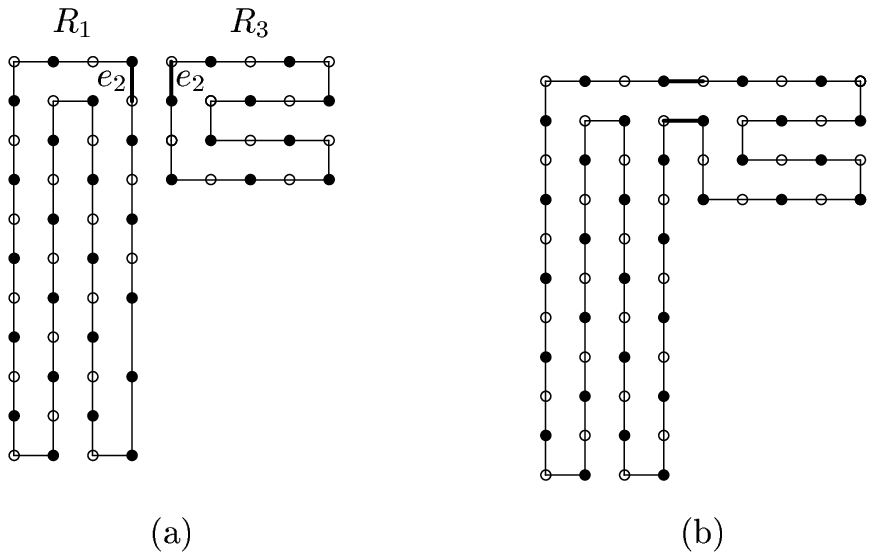}
  \caption[]%
  {\small Combining two Hamiltonian cycles.}
  \label{cm}
\end{figure}

\begin{figure}[htb]
  \centering
  \includegraphics[scale=1]{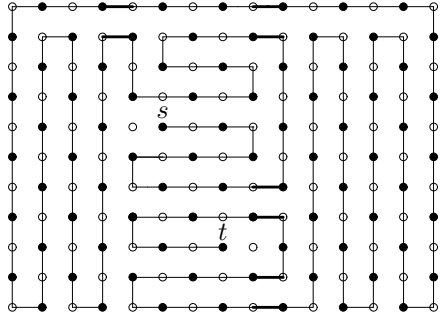}
  \caption[]%
  {\small The longest path between $s$ and $t$.}
  \label{cn}
\end{figure}

Considering all of the above, we get the algorithm of finding a
longest path in rectangular grid graphs, as shown in Algorithm
\ref{alg:1}.
\newcommand{\LET}{\STATE \textbf{let }}
\newcommand{\SET}{\STATE \textbf{set }}
\newcommand{\PROC}{\STATE \textbf{procedure }}
\newcommand{\RET}{\STATE \textbf{return }}
\begin{algorithm}
\caption{The longest path algorithm} \label{alg:1}
\begin{algorithmic}
\PROC LongestPath$(R(m,n),s,t)$
\end{algorithmic}
\begin{algorithmic}
{\small \STATE \textbf{Step 1.} By a peeling operation, $R(n,m)$
partitions into five disjoint rectangular grid subgraphs $R_{1}$ to
$R_{5}$, such that $s,t\in R_{5}$ \STATE \textbf{Step 2.} Finding
longest path between $s$ and $t$ in $R_{5}$. 
\STATE \textbf{Step 3.} Construct Hamiltonian cycles in rectangular
grid subgraphs $R_{1}$ to $R_{4}$.
 \STATE \textbf{Step 4.} Construct a longest
path between $s$ and $t$ by combine all Hamiltonian cycles and a
longest path.}
\end{algorithmic}
\end{algorithm}

Consider the pseudo-code of our algorithm in Algorithm \ref{alg:1}.
The step 1 dose only a constant number of partitioning, during the
peeling operation, which is done in constant time. The step 2
trisects $R_5$ which requires also a constant number of
partitioning. Then finds a longest path in $R_5$ by merging paths of
the partitions which can be done in linear time. The step 3 finds
Hamiltonian cycles of $R_1$ to $R_4$ which is done in linear time.
The step 4 which combines the Hamiltonian cycles and the longest
path requires only constant time. Therefore, in total our sequential
algorithm has linear-time complexity.
\section{The parallel algorithm}
In this section, we present a parallel algorithm for the  longest
path problem. This algorithm is based on the sequential algorithm
presented in the previous sections. Our parallel algorithm runs on
every parallel machine, we do not need any inter-processor
connection in our algorithm. We assume there are $nm$ processors and
they work in SIMD mode. For simplicity, we use a two-dimensional
indexing scheme. Each vertex $v$ of the given rectangular grid graph
$R(m,n)$ is mapped to processor $(v_x,v_y)$. Each processor knows
its index, coordinates $s$ and $t$, and $m$ and $n$. \par The
peeling phase is parallelized easily, every processor calculates the
following four variables, in parallel \cite{CST:AFAFCHPIM}:

 $r_{1}=
  \begin{cases}
  s_{x} - 2;   &  \text s_{x}\bmod \ 2=0 \\
    s_{x} - 1;& \text otherwise
  \end{cases}$
  \\

$r_{2}=
  \begin{cases}
  t_{x} +1;   &  \text t_{x}\bmod \ 2=m\bmod\ 2 \\
    t_{x} +2;& \text otherwise
  \end{cases}$
  \\

$r_{3}=
  \begin{cases}
  min(s_{y} ; t_{y})-2; & \text min(s_{y} ; t_{y})\bmod\ 2 = 0 \\
   min(s_{y} ; t_{y})-1 ;& \text  otherwise
  \end{cases}$
  \\

$r_{4}=
  \begin{cases}
  max(s_{y} ; t_{y})+1; & \text max(s_{y} ; t_{y})\bmod\ 2 = n \bmod\ 2 \\
   max(s_{y} ; t_{y})+2 ;& \text  otherwise
  \end{cases}$
  \\
  \par
Where variables $r_1$, $r_2$, $r_3$ and $r_4$ correspond to the
right-most column number of $R_1$,  the left-most column number of
$R_2$, the bottom row number of $R_3$, and the top row number of
$R_4$, respectively. Then a processor can identify its
subrectangular by comparing its coordinates with these four
variables. In case (F1$^{'}$) and (F2$^{'}$), the boundary
adjustment can be done by simply decrementing $R_{1}$, $R_{2}$,
$R_{3}$ or $R_{4}$ or incrementing $R_{3}$ or $R_{4}$.

\par The trisecting phase is also parallelized in a similar manner. In the following we describe
how we parallelized the horizontal trisecting, in two cases when
$R(m,n)$ is even$\times$odd (or odd$\times$odd) and when it is
even$\times$even. In case $R(m,n)$
is even$\times$odd or odd$\times$odd, every processor simultaneously calculate the following two variables:\\
$l=
  \begin{cases}
  min(s_{y} ; t_{y}) \ & \text  min(s_{y} ; t_{y})\bmod\ 2 =
  0\\
   min(s_{y} ; t_{y})+1 & \text  min(s_{y} ; t_{y})\bmod\ 2 \neq
  0
  \end{cases}$\\
$r=
  \begin{cases}
  max(s_{y} ; t_{y}) \ & \text  max(s_{y} ; t_{y})\bmod\ 2 =
  0\\
  max(s_{y} ; t_{y})-1 & \text  max(s_{y} ; t_{y})\bmod\ 2 \neq
  0
  \end{cases}$\par
Where variables $l$ and $r$ correspond to the bottom row number of
$R^{s}_{5}$ (resp. $R^{t}_{5}$), and the top row number of
$R^{t}_{5}$ (resp. $R^{s}_{5}$), respectively.\par In case $R(m,n)$
is even$\times$even, every processor simultaneously calculate the
following two variables:
\\
$l=
  \begin{cases}
  min(s_{y} ; t_{y}) \ & \text  min(s_{y} ; t_{y})\bmod\ 2 =
  0\\
   min(s_{y} ; t_{y})+1 & \text  min(s_{y} ; t_{y})\bmod\ 2 \neq
  0
  \end{cases}$\\
$r=
  \begin{cases}
  max(s_{y} ; t_{y}) \ & \text  max(s_{y} ; t_{y})\bmod\ 2  \neq
  0\\
  max(s_{y} ; t_{y})-1 & \text  max(s_{y} ; t_{y})\bmod\ 2=
  0
  \end{cases}$\\

A similar method can be used to parallelize the vertically
trisecting.\\
After peeling and trisecting, all processors in the same
subrectangles simultaneously construct either a longest path,
Hamiltonian path or cycle according to the pattern associated with
the subrectangle. For constructing a Hamiltonian path in a
rectangular grid graph, we use the constant-time algorithm of
\cite{CST:AFAFCHPIM}. For constructing a Hamiltonian cycle in an
even-sized rectangle, we use the constant-time algorithm of
\cite{SST:EPPRAOM} in which every processor computes its successor
in the cycle. This algorithm is given in Algorithm \ref{alg:2}; see
Figure \ref{e}(a) .\par For constructing a longest path, parallel
algorithms can be easily developed for each different pattern shown
in  Figure \ref{e}(b), (c). As two examples, for constructing a
longest path between vertices $(2,1)$ and $(m-1,n)$ in an
odd$\times$odd rectangular grid graph $R(m,n)$, and vertices
$(m-1,1)$ and $(n-1,1) $ in an even$\times$even rectangular grid
graph $R(m,n)$. We have developed the simple algorithms Algorithm
\ref{alg:3} and \ref{alg:4}, respectively. The algorithms for other
patterns can be derived in the similar way. \par Then combining
phase is parallelized as follows.  The two processors at the two
endpoints of a corner edge in a Hamiltonian cycle $c_1$ check
whether a neighboring Hamiltonian cycle $c_2$ exists or not. If
$c_2$ exists, then their successors are changed to the adjacent
processors in $c_2$. Similarly, the two processors at the endpoints
of a corner edge in the longest path $P$ in $R_5$ also check the
existence of the adjacent edge in the Hamiltonian cycle $C$, and
change their successors. Thus, the combining phase can be
parallelized in constant steps without inter-processor
communication.

\begin{figure}[htp]
  \centering
  \includegraphics[scale=1]{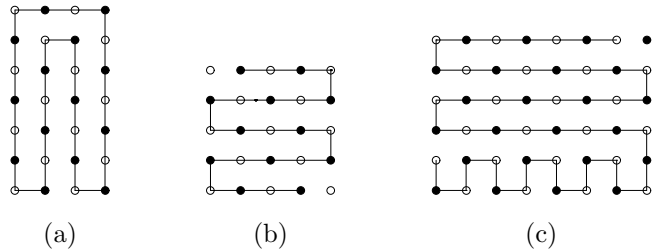}
  \caption[]%
  {(a) A Hamiltonian cycle in $R(4, 7)$, (b) and (c) two patterns of longest path in $R_5(5,5)$ and $R_5(8,6)$.}
  \label{e}
\end{figure}
\begin{algorithm}
\caption{The Hamiltonian cycle parallel algorithm for an even-sized
rectangular grid graphs} \label{alg:2}
\begin{algorithmic}
\PROC LongestPath$R(m,n)$
\end{algorithmic}
\begin{algorithmic}[1]
{\small \STATE \textbf{for} each processor $(x,y)$ in $R(m,n)$
\textbf{do} in parallel \STATE \textbf{if} $y=1$, \textbf{then}
successor $(x,y)$ $\leftarrow$ $(x+1,y)$ \STATE \textbf{elseif}
($y=2$, $x$ is odd and $x\neq 1$) or ($y=n$ and $x$ even),
\textbf{then} successor $(x,y)$ $\longleftarrow$ $(x-1,y)$\STATE
\textbf{slesif} $x$ is even and $y<n$, \textbf{then} successor
$(x,y)$ $\longleftarrow$ $(x,y+1)$
 \STATE \textbf{else} $x$ is odd and $y\leq n$, \textbf{then} successor $(x,y)$ $\longleftarrow$ $(x,y-1)$}
\end{algorithmic}
\end{algorithm}
\begin{algorithm}
\caption{The longest path parallel algorithm for odd$\times$odd
rectangular grid graphs} \label{alg:3}
\begin{algorithmic}
\PROC LongestPath$(R_5(m_5,n_5),s,t)$
\end{algorithmic}
\begin{algorithmic}[1]
{\small \STATE \textbf{for} each processor $(x,y)$ in $R_5(m_5,n_5)$
\textbf{do} in parallel \STATE \textbf{if} $x=1$ and $y=1$ or $x=m$
and $y=n$, \textbf{then} successor $(x,y)$ $\longleftarrow$ nill
\STATE \textbf{elseif} $y$ is odd and $x<m$, \textbf{then} successor
$(x,y)$ $\longleftarrow$ $(x+1,y)$\STATE \textbf{slesif} $y$ is odd
and $x=m$, \textbf{then} successor $(x,y)$ $\longleftarrow$
$(x,y+1)$\STATE \textbf{elseif} $y$ is even and $x>1$, \textbf{then}
successor $(x,y)$ $\longleftarrow$ $(x-1,y)$
 \STATE \textbf{else} $y$ is even and $x=1$, \textbf{then} successor $(x,y)$ $\longleftarrow$ $(x,y+1)$}
\end{algorithmic}
\end{algorithm}

\begin{algorithm}
\caption{The longest path parallel algorithm for even$\times$even
rectangular grid graphs} \label{alg:4}
\begin{algorithmic}
\PROC LongestPath$(R_5(m_5,n_5),s,t)$
\end{algorithmic}
\begin{algorithmic}[1]
{\small \STATE \textbf{for} each processor $(x,y)$ in $R_5(m_5,n_5)$
\textbf{do} in parallel \STATE \textbf{if} $x=m$ and $y=1$,
\textbf{then} successor $(x,y)$ $\longleftarrow$ nill \STATE
\textbf{elseif} ($y$ is odd and $x=m$), ($y$ is even and $x=1$) or
($y=n$ and $x$ is even)  \textbf{then} successor $(x,y)$
$\longleftarrow$ $(x,y-1)$\STATE \textbf{slesif} ($y$ is odd and
$x<m$), ($y=n-1$ and $x$ is even), ($y=n$ and $x$ is odd)
\textbf{then} successor $(x,y)$ $\longleftarrow$ $(x+1,y)$
 \STATE \textbf{elseif} $y$ is even and $x>1$ \textbf{then} successor $(x,y)$ $\longleftarrow$ $(x-1,y)$
 \STATE \textbf{elseif} $y=n-1$ and $x$ is odd \textbf{then} successor $(x,y)$ $\longleftarrow$ $(x,y+1)$
 }

\end{algorithmic}
\end{algorithm}

\section{Conclusion and future work} \label{ConclusionSect}
We presented a linear-time sequential algorithm for finding a longest path in a
rectangular grid graph between any two given vertices. Since the
longest path problem is NP-hard in general grid graphs
\cite{IPS:HPIGG}, it remains open if the problem is polynomially
solvable in solid grid
graphs. Based on the sequential algorithm a constant-time parallel algorithm
is introduced for the problem, which can be run on every parallel machine.\\[1cm]


\begin{thebibliography}{00}

\bibitem{BH:FAPOSL}
A. Bj\"{o}rklund and T. Husfeldt, Finding a path of superlogarithmic
length, SIAM J. Comput., 32(6):1395-1402, 2003.

\bibitem{BSZVGF:OCALPIAT}
R. W. Bulterman, F. W. van der Sommen, G. Zwaan, T. Verhoeff, A. J.
M. van Gasteren and W. H. J. Feijen, On computing a longest path in
a tree, Information Processing Letters, 81(2):93-96, 2002.

\bibitem{CST:AFAFCHPIM}
S. D. Chen, H. Shen and R. Topor, An efficient algorithm for
constructing Hamiltonian paths in meshes, J. Parallel Computing,
28(9):1293-1305, 2002.

\bibitem{SST:EPPRAOM}
S. D. Chen, H. Shen, R. W. Topor, Efficient parallel
permutation-based range-join algorithms on meshconnected computers,
in: Proceedings of the 1995 Asian Computing Science Conference,
Pathumthani, Thailand, Springer-Verlag, 225-238, 1995.

\bibitem{D:GT}
R. Diestel, Graph Theory, Springer, New York, 2000.

\bibitem{GN:FLPCAC}
H. N. Gabow and S. Nie, Finding long paths, cycles and circuits,
19th annual International Symp. on Algorithms and Computation
(ISAAC), LNCS, 5369:752-763, 2008.

\bibitem{GJ:CAI}
M. R. Garey and D. S. Johnson, Computers and Intractability: A Guide
to the Theory of NP-completeness, Freeman, San Francisco, 1979.

\bibitem{vyf:hpotgg}
{V. S. Gordon, Y. L. Orlovich and F. Werner, Hamiltonian properties
of triangular grid graphs, \em Discrete Math.}, 308 (2008)
6166-6188.

\bibitem{G:FALPIACMD}
G. Gutin, Finding a longest path in a complete multipartite digraph,
SIAM J. Discrete Math., 6(2):270-273, 1993.

\bibitem{Imnrx:hcihgg}
{K. Islam, H. Meijer, Y. Nunez, D. Rappaport and H. xiao,
Hamiltonian Circuts in Hexagonal Grid Graphs, \em CCCG}, (2007)
20-22.
\bibitem{IPS:HPIGG}
A. Itai, C. Papadimitriou and J. Szwarcfiter, Hamiltonian paths in
grid graphs, SIAM J. Comput., 11(4):676-686, 1982.


\bibitem{KMR:OATLPIAG}
D. Karger, R. Montwani and G. D. S. Ramkumar, On approximating the
longest path in a graph,  Algorithmica, 18(1):82-98, 1997.

\bibitem{FAA:ALAFFLPIRGG}
F. Keshavarz-Kohjerdi, A. Bagheri and A. Asgharian-Sardroud, A
Linear-time Algorithm for the Longest Path Problem in Rectangular
Grid Graphs, {\em Discrete Applied Math.}, 160(3): 210-217, 2012.

\bibitem{kb:hpiscogg}
{F. Keshavarz-Kohjerdi and A. Bagheri, Hamiltonian Paths in Some
Classes of Grid Graphs},  {\em Journal of Applied Mathematics}, accepted.

\bibitem{LU:HCISGG}
W. Lenhart and C. Umans, Hamiltonian Cycles in Solid Grid Graphs,
Proc. 38th Annual Symposium on Foundations of Computer Science (FOCS
'97), 496-505, 1997.

\bibitem{LMN:TLPPIPOIG}
K. Loannidou, G. B. Mertzios and S. Nikolopoulos, The longest path
problem is polynomial on interval graphs, Proc. of 34th Int. Symp.
on Mathematical Foundations of Computer Science, Springer-Verlag,
Novy Smokovec, High Tatras, Slovakia, 5734:403-414, 2009.

\bibitem{LM:HPOARC}
F. Luccio and C. Mugnia, Hamiltonian paths on a rectangular
chessboard, Proc. 16th Annual Allerton Conference, 161-173, 1978.

\bibitem{MC:ASPAFTLPPOCG}
G. B. Mertzios and D. G. Corneil, A simple polynomial algorithm for
the longest path problem on Cocomparability Graphs, J. Comput. Sci.,
Submitted 2010.

\bibitem{M}
{B. R. Myers, Enumeration of tours in Hamiltonian rectangular latice
graphs, \em Mathematical Magazine}, 54(1) (1981), 19-23.

\bibitem{6}
{M. Nandi, S. Parui and A. Adhikari, The domination numbers of
cylindrical grid graphs, \em Applied Mathematics and Computation},
217(10) (2011) 4879-4889.

\bibitem{AEB:SSAG}
A. N. M. Salman, Contributions to Graph Theory, Ph.D. Thesis,
University of Twente, (2005).

\bibitem{UU:OCLPISGC}
R. Uehara and Y. Uno, On Computing longest paths in small graph
classes, Int. J. Found. Comput. Sci., 18(5):911-930, 2007.

\bibitem{CT:HPOGG}
C. Zamfirescu and T. Zamfirescu, Hamiltonian Properties of Grid
Graphs, {\em  SIAM J. Math.}, 5(4) (1992) 564-570.

\bibitem{ZL:AFLPIG}
Z. Zhang and H. Li, Algorithms for long paths in graphs, Theoretical
Comput. Sci., 377(1-3):25-34, 2007.

\bibitem{wqz}
W. Q. Zhang and Y. J. Liu, Approximating the longest paths in grid
graphs, Theoretical Computer Science, 412(39): 5340-5350, 2011.

\end{thebibliography}
\end{document}